\documentclass[11pt]{article}
\usepackage[utf8]{inputenc}
\usepackage{style}
\usepackage{thm-restate}
\usepackage{thmtools}
\usepackage[margin=1in]{geometry}
\usepackage{tikz}
\usepackage{verbatim}
\usepackage{cancel}
\usepackage{titling}
\usepackage{multirow}
\usepackage{cancel}
\usepackage{quantikz}
\usepackage{bm}
\usepackage{enumitem}
\usepackage{mathtools}

\newcommand{\ignore}[1]{}

\definecolor{cb-salmon-pink}{RGB}{255, 182, 119}
\crefname{enumi}{Step}{Steps}
\usepackage{float}
\newcommand{\Tr}{\mathrm{Tr}}

\title{Few Single-Qubit Measurements Suffice\\to Certify Any Quantum State}
\author{Meghal Gupta\thanks{UC Berkeley, \texttt{meghal@berkeley.edu}}\and William He\thanks{Carnegie Mellon University, \texttt{wrhe@cs.cmu.edu}}\and Ryan O'Donnell\thanks{Carnegie Mellon University, \texttt{odonnell@cs.cmu.edu}. Supported in part by a grant from Google Quantum AI.}}
\begin{document}
\allowdisplaybreaks
\maketitle
\begin{abstract}
   
    A fundamental task in quantum information science is \emph{state certification}: testing whether a lab-prepared $n$-qubit state is close to a given hypothesis state. In this work, we show that \emph{every} pure hypothesis state can be certified using only $O(n^2)$ single-qubit measurements applied to $O(n)$ copies of the lab state. Prior to our work, it was not known whether even subexponentially many single-qubit measurements could suffice to certify arbitrary states.
    This resolves the main open question of Huang, Preskill, and Soleimanifar (FOCS 2024, QIP 2024).

    Our algorithm also showcases the power of \emph{adaptive measurements}: within each copy of the lab state,  previous measurement outcomes dictate how subsequent qubit measurements are made. We show that the adaptivity is necessary, by proving an exponential lower bound on the number of copies needed for any nonadaptive single-qubit measurement algorithm.  
\end{abstract}
\DeclarePairedDelimiterX\diverg[2]{(}{)}{#1 \mathrel{}\mathclose{}\delimsize{:}\mathopen{}\mathrel{} #2}
\newcommand{\Del}[3]{\Delta_{#1}\diverg{#2}{#3}}
\newcommand{\EFid}[3]{\mathrm{EFid}_{#1}\diverg{#2}{#3}}

\newcommand{\ketbb}{\ket{\smash{b^\bot}}}
\newcommand{\tpO}{\smash{\widetilde{\smash{\psi}}^0}}
\newcommand{\ktpO}
{\ket{\tpO}}
\newcommand{\tpI}{\smash{\widetilde{\smash{\psi}}^1}}
\newcommand{\ktpI}
{\ket{\tpI}}
\newcommand{\tfO}{\smash{\widetilde{\phi}^0}}
\newcommand{\ktfO}
{\ket{\tfO}}
\newcommand{\tfI}{\smash{\widetilde{\phi}^1}}
\newcommand{\ktfI}
{\ket{\tfI}}
\newcommand{\psip}{\smash{\psi'}}
\newcommand{\phip}{\smash{\phi'}}
\newcommand{\tarp}{\smash{\overline{\tar}}}
\newcommand{\labp}{\smash{\overline{\lab}}}
\newcommand{\tarz}{\smash{\tar^{0}}}
\newcommand{\taro}{\smash{\tar^{1}}}
\newcommand{\labz}{\smash{\lab^{0}}}
\newcommand{\labo}{\smash{\lab^{1}}}

\section{Introduction}
A fundamental task in quantum information science is to test whether an unknown $n$-qubit state $\rho_{\lab}$ prepared in the lab is equal (or close) to a known target state~$\ket{\tar}$. This task is known as \emph{quantum state certification}. It is essential for benchmarking quantum devices, validating the outcomes of quantum experiments, and verifying the correctness of practical implementations of quantum algorithms and protocols. See e.g.\ \cite{buadescu2019quantum, zhu2019statistical, kliesch2021theory, huang2024certifying} for in-depth discussions of the importance and applications of quantum state certification. 

More formally, we are given a classical description of a known pure state $\ket{\tar}$, along with copies of an unknown (possibly mixed) lab state $\rho_{\lab}$. Our goal is to decide whether the fidelity $\bra{\tar}\rho_{\lab}\ket{\tar}$ is close to $1$, in which case we \textsc{Accept}, or significantly smaller than $1$, in which case we \textsc{Reject}.

If we ignore measurement complexity, the optimal-copy-complexity approach is straightforward: measure each copy of $\rho_{\lab}$ using the POVM $\{\ketbra{\tar}{\tar}, \Id - \ketbra{\tar}{\tar}\}$, and accept if the first outcome occurs frequently enough. This requires only $\Theta(1/\epsilon)$ copies to distinguish $\bra{\tar}\rho_{\lab}\ket{\tar} < 1 - \epsilon$ from $\bra{\tar}\rho_{\lab}\ket{\tar} = 1$ (or even $\geq 1 - \epsilon/2$).

However, actually implementing this $n$-qubit measurement is roughly as hard as preparing the state~$\ket{\tar}$. Given that a main use of state certification is to verify that we can in fact prepare $\ket{\tar}$, it is almost circular to allow measurement of the POVM $\{\ketbra{\tar}{\tar}, \Id - \ketbra{\tar}{\tar}\}$. What we would  actually like is a much simpler algorithm for certifying a state --- ideally something that's polynomially complex even when preparing $\ket{\tar}$ is exponentially complex (as it is for most states). These considerations motivate certification algorithms that use only single-qubit measurements: This would allow for efficient certification even when preparation is more complex and unreliable.

Prior to this work, it was unknown whether all pure states $\ket{\tar}$ could be certified using only few single-qubit measurements.\footnote{We remark that if $\ket{\tar}$ is allowed to be a mixed state, it is known that certification in general requires exponentially many copies, even when arbitrarily complex measurements are allowed~\cite{CHW07,swan1swan2_21}.} Several algorithms had been developed for this task, but each required some significant concession. For example, some algorithms require exponentially many copies of $\rho_{\lab}$: \cite{flammia2011direct,da2011practical, aolita2015reliable}. Others apply only to restricted classes of hypothesis states or restricted noise models:~\cite{hayashi2006study,flammia2011direct,aolita2015reliable,hayashi2015verifiable,morimae2016quantum,morimae2017verification,takeuchi2018verification,gluza2018fidelity,hayashi2019verifying,li2019efficient,liu2019efficient,yu2019optimal,zhu2019efficient,li2021verification,huang2024certifying}. 

Indeed, given each of these prior results comes with caveats, it was natural to suspect that fully general certification using only polynomially many single-qubit measurements might be impossible. Even after Huang, Preskill, and Soleimanifar~\cite{huang2024certifying} gave such an algorithm for Haar-random target states $\ket{\tar}$ (which are typically highly entangled), they left it open to identify an explicit state $\ket{\tar}$ for which certification with single-qubit measurements requires super-polynomially many copies.

Perhaps surprisingly, this work shows that there is no such state: Efficient certification with single-qubit measurements is possible for \emph{all} quantum states. We present a simple, general quantum state certification algorithm that works for every pure hypothesis state $\ket{\tar}$, and which uses only single-qubit measurements applied to $O(n/\epsilon)$ copies of $\rho_{\lab}$:

\begin{maintheorem}\label{thm:main with repetitions}
    There exists an algorithm that given parameters $0 < \eps, \delta < 1$, oracle access to the classical description of a pure state $\ket{\tar}$ (via the model in \Cref{def:access model}), and $O(n\epsilon^{-1}\ln(1/\delta))$ copies of $\rho_{\lab}$, makes only single-qubit measurements to the copies of $\rho_{\lab}$ and outputs:
    \begin{itemize}
        \item \textsc{Accept} with probability at least $1-\delta$ if $\bra{\tar}\rho_{\lab}\ket{\tar} \geq 1 - \frac{\epsilon}{2n}$.
        \item \textsc{Reject} with probability at least $1-\delta$ if $\bra{\tar}\rho_{\lab}\ket{\tar}\leq 1-\eps$.
    \end{itemize}
    Moreover, the algorithm runs in time linear in the number of measurements.\vspace{1ex}
\end{maintheorem}

This theorem follows directly from \Cref{thm:main} below, together with repetition and a Chernoff bound:
\begin{theorem}\label{thm:main}
    There exists an algorithm that given oracle access to the classical description of a pure state $\ket{\tar}$ via the model in \Cref{def:access model} and \emph{one} copy of a mixed state $\rho_{\lab}$, makes only single-qubit measurements to $\rho_{\lab}$ and outputs: 
    \begin{itemize}
        \item \textsc{Accept} with probability at least $\text{Fid} \coloneqq \bra{\tar}\rho_{\lab}\ket{\tar}$;
        \item \textsc{Reject} with probability at least $\text{Infid}/n \coloneqq (1-\bra{\tar}\rho_{\lab}\ket{\tar})/n$.
    \end{itemize}
    Moreover, the algorithm runs in time linear in the number of measurements.
\end{theorem}

Our algorithm also shows the power of \emph{adaptive} measurements.  For each copy of $\rho_{\lab}$, how our algorithm measures the next qubit sometimes depends on the outcomes of the previous measurements (though our algorithm is nonadaptive across copies).  This is in contrast to the fully non-adaptive algorithm appearing in~\cite{huang2024certifying}. 
The use of adaptivity turns out to be necessary; we show an exponential copy-complexity lower bound against algorithms that are non-adaptive or only use a small amount of adaptivity.

To be precise, we say that an algorithm is $\ell$-adaptive if for each copy of $\rho_{\lab}$ it receives, there is some $S\subseteq[n]$ with $|S|\leq \ell$ such that it measures the qubits in $[n]\setminus S$ in a product basis $\bigotimes_{i\in[n]\setminus S}\calB_i$ (each $\calB_i = \{\ket{\smash{b_i}}, \ket{\smash{b_i^\bot}}\}$ a $1$-qubit basis) and measures the rest of the qubits in some arbitrary basis, potentially depending on the previous outcomes.  We allow the algorithm to act adaptively across copies; that is, the product basis it uses for the $t$th copy may depend on measurement outcomes from the preceding copies.
\begin{theorem}\label{prop:lower bound}
    There exist $c_1, c_2, c_3 ,c_4> 0$ such that the following holds. For $n\geq1$, there exists an $n$-qubit pure state $\ket{\tar}$ and an $n$-qubit mixed state $\rho_{\lab}$ such that: $\bra{\tar}\rho_{\lab}\ket{\tar}\leq 2^{-c_1n}$, but any $c_2n$-adaptive certification algorithm making that succeeds with probability at least $1/2 + 2^{-c_3 n}$ in distinguishing $\ketbra{\tar}{\tar}$ from $\rho_{\lab}$ must use at least $2^{c_4 n}$ copies of~$\rho_{\lab}$. 
\end{theorem}
We prove \Cref{prop:lower bound} in \Cref{sec:lower bound}.

\paragraph{Future directions.} Several natural questions remain open for future work. We state a few here:

First, a basic information-theoretic argument shows that any certification algorithm for a pure target state $\ket{\tar}$ must use at least $\Omega(1/\epsilon)$ copies of $\rho_{\lab}$, even without restricting to single-qubit measurements. Can stronger lower bounds be proven for algorithms limited to single-qubit measurements? In particular, our algorithm uses $O(n/\epsilon)$ copies --- is this dependence on $n$ necessary?

Second, our algorithm is only $1/n$-tolerant: To guarantee $\rho_{\lab}$ is accepted with high probability, it needs to satisfy $\bra{\tar}\rho_{\lab}\ket{\tar} \geq 1 - \Omega(\epsilon/n)$, rather than $1-\Omega(\eps)$. Can this dependence on $n$ be improved or even eliminated? For instance, is it possible to design a certification algorithm with similar complexity that distinguishes between $\bra{\tar}\rho_{\lab}\ket{\tar} \geq 1 - \epsilon/2$ and $\bra{\tar}\rho_{\lab}\ket{\tar} \leq 1 - \epsilon$?

\Cref{prop:lower bound} shows that adaptivity is required to certify a worst-case state with single-qubit measurements. Unfortunately, adaptivity is a significant downside for practical implementations. In contrast to our lower bound, Huang et al.~\cite{huang2024certifying} showed that for an average-case state, adaptivity is not necessary. Can we classify the set of states that require adaptivity to certify? Moreover, for states that do require adaptivity to certify, how much is necessary? Our algorithm requires $n$ rounds of adaptivity, but perhaps this can be reduced.

%





\section{The Algorithm}

In this section we prove \Cref{thm:main}.

\subsection{DT Bases}
At a high level, our algorithm will pick a random $k\in [n]$, measure the first $k-1$ qubits in the computational basis, then measure the last $n-k$ qubits in some carefully constructed basis, and finally measure the $k$th qubit. We begin with some definitions and facts that will help us understand the carefully constructed basis in which the last $n-k$ qubits are measured.
\begin{definition}
    Let $\calT$ be a depth-$n$ binary tree~$\calT$ in which each internal node's two outgoing edges are labeled by orthogonal single-qubit states. 
    By a slight abuse of notation, we identify $\calT$ with the orthonormal basis $\{\ket{\ell} : \ell \text{ a leaf in } \calT\}$, and call $\calT$ a \emph{decision tree (DT) basis} for $(\C^2)^{\otimes n}$. 
\end{definition}
\begin{remark}
    Note that given an $n$-qubit DT basis $\calT$, one can measure an $n$-qubit state in this basis by using adaptive single-qubit measurements.
    Indeed, a deterministic adaptive algorithm that measures qubits in the order $1 \dots n$ is \emph{equivalent} to a DT basis.
    We also comment that although $\calT$ is a \emph{basis of product states}, this is a more general object than a \emph{product basis}, the latter being what a non-adaptive algorithm uses to measure each copy.
\end{remark}

\begin{definition}
    Recall that qudit $\ket{\phi} \in \C^d$ is said to be a \emph{phase state} in basis $\mathcal{L} = \{ \ket{1}, \dots, \ket{d} \}$ if it can be written as
    \begin{equation*}
        \ket{\phi} = \frac{1}{\sqrt{d}} \sum_{\ell=1}^d \omega_\ell \ket{\ell} \qquad \text{for some phases $\omega_\ell$}.
    \end{equation*}
    In other words, measuring in basis $\cal{L}$ yields each outcome with equal probability.
\end{definition}

\begin{lemma}\label{lem:perpy}
    Let $\rho^0$ and $\rho^1$ be single-qubit \emph{mixed} states. Then there exists a $1$-qubit basis $\calB = \{\ket{b}, \ketbb\}$ such that for $i = 0, 1$, measuring $\rho^i$ in basis $\calB$ yields each outcome $\ket{b}$ and $\ketbb$ with probability~$\frac12$.
\end{lemma}
\begin{proof}
     It suffices to choose qubit $\ket{b}$ so that its Bloch sphere representation is orthogonal to that of $\rho^0, \rho^1$ (this choice is unique up to sign unless $\rho^0, \rho^1$ are collinear in the Bloch ball).  For such a $\ket{b}$, and its Bloch-antipode $\ketbb$, 
     measuring $\rho^i$ with respect to~$\{\ket{b}, \ketbb\}$ indeed gives each outcome with probability~$\frac12$.
\end{proof}


\begin{corollary} \label{cor:perpy}
     Let $\ket{\phi^0}$ and $\ket{\phi^1}$ be $m$-qubit quantum states. Then there exists a decision tree basis~$\calT$ for $(\C^2)^{\otimes n}$ in which they are both phase states.\footnote{This corollary can also be obtained from the algorithm in~\cite[Sec.~III]{zhou2020saturating}, by taking its matrix $\tilde{M}$ to be $(1+\mathfrak{i})\ketbra{\phi^0}{\phi^0} - \ketbra{\phi^1}{\phi^1} - \mathfrak{i} \cdot \Id/2^n$.}
\end{corollary}
\begin{proof}
    For $k = 1 \dots m$, we inductively build the labels on the edges at depth at most $k$ so that for both $i = 0, 1$ and all depth-$k$ nodes~$w$, the probability of obtaining outcome~$w$ when measuring $\ket{\phi^i}$ with the tree (so far) is~$2^{-k}$.
    For the base case of $k = 1$, we apply \Cref{lem:perpy}; the resulting $1$-qubit basis $\calB = \{\ket{b}, \ketbb\}$ serves as the edge labels out from the root of $\calT$.
    Now to extend from $k$ to $k+1$ we need, for each node $w$ at depth~$k$, a $1$-qubit basis $\calB_w = \{\ket{b_w}, \ket{b_w^\bot}\}$ such that measuring each of the reduced states $\ket{(\phi^0)^{w}}, \ket{(\phi^1)^{w}}$ in basis~$\calB_w$ yields each outcome with probability~$\frac12$. This may be obtained by applying \Cref{lem:perpy} to $\ket{(\phi^0)^{w}}, \ket{(\phi^1)^{w}}$.
\end{proof}

\subsection{Our Algorithm and Its Key Subtest}
\newcommand{\tarline}{\ket{\tar'}}
\newcommand{\labline}{\ket{\lab'}}
Let $\ket{\tar}$ be the $n$-qubit state to be certified. In our proof of \Cref{thm:main}, we may assume without loss of generality that the lab state is a pure one,~$\ket{\lab}$. This is because we can regard $\rho_{\lab}$ as a probability distribution over pure states, since our algorithm may only perform measurements to a single copy of lab. (This equivalence can be shown by linearity of expectation.)
Also, for expositional simplicity, in this section we regard $\ket{\tar}$ as completely ``known''; we will not be concerned with the complexity of interacting with $\ket{\tar}$.  The straightforward details of the access model we actually assume, and the running time, are deferred to \Cref{sec:access model}.

We will use a collection of DT bases $\calT_x$, for $x\in \{0,1\}^{\leq n}$.
For all $1\leq k\leq n$ and all binary strings $x\in\{0,1\}^{k-1}$, let $\ket{\tar^x}$ and $\ket{\lab^x}$ be the reduced states of $\ket{\tar}$ and $\ket{\lab}$, respectively, after measuring the first $k-1$ qubits of these states in the computational basis and observing $x$. For all $x\in\{0,1\}^{\leq n}$, define the DT basis $\calT_x$ so that in this basis, both $\ket{\tar^{x0}}$ and $\ket{\tar^{x1}}$ are phase states; we know this DT basis exists by \Cref{cor:perpy}.

We may now state our algorithm:

\begin{algorithm}[H]
    \vspace{0.3em}
    \textbf{Input:} Full classical description of $\ket{\tar}$ and one copy of an unknown state $\ket{\lab}$. \\
    \textbf{Output:} \textsc{Accept} or \textsc{Reject}.
    \begin{algorithmic}[1]
    \State \textbf{Sample} $k\in[n]$ uniformly at random.
    \State \textbf{Measure} the first $k-1$ qubits of $\ket{\lab}$ in the computational basis, obtaining~$x\in\{0,1\}^{k-1}$.    
    \State \textbf{Measure} the last $n-k$ qubits of $\ket{\lab^x}$ in the basis $\calT_x$, obtaining~$\ell$.
    \State Let $\labline$ be the resulting $1$-qubit state; let $\tarline$ denote $\ket{\tar}$ conditioned on outcomes $x, \ell$.
    
    \noindent \textbf{Measure} $\labline$ in a basis containing $\tarline$, 
    and \textbf{Output} \textsc{Accept} iff the outcome is $\tarline$.
    \end{algorithmic}
    \caption{\textsc{Certify}$(\ket{\lab}, \ket{\tar})$}
    \label{alg:main}
\end{algorithm}

\begin{definition}
    We refer to Steps 3--4 of our algorithm as the \emph{subtest} performed on the \mbox{$(n-k)$}-qubit states $\ket{\tar^x}$ and $\ket{\lab^x}$, with DT basis $\calT_x$.  We will also use the notation $\textsc{SubTest}_{\calT_x}(\ket{\lab^x}: \ket{\smash{\tar^x}})$.
\end{definition}

The key to analyzing the $\textsc{Certify}$ algorithm is to compare the subtest acceptance probability 
to the following quantity:

\begin{definition}
    For $(n-k+1)$-qubit states $\ket{\lab^x}, \ket{\tar^x}$ as in $\textsc{Certify}$, 
    we define their \emph{fidelity gap} to be
    \begin{equation*}
        \Delta(\ket{\lab^x}: \ket{\tar^x}) \coloneqq \mathop{\E}_{b}\left[\abs{\braket{\tar^{xb}}{\lab^{xb}}}^2\right]- \abs{\braket{\tar^x}{\lab^x}}^2.
    \end{equation*}
    Here the expectation is over the random outcome $b \in \{0,1\}$ obtained when measuring the first qubit of $\ket{\lab^x}$ in the computational basis. 
\end{definition}
\begin{remark} \label{rem:nonneg}
    $\Delta(\ket{\lab^x}: \ket{\tar^x}) \geq 0$ always. This follows from the data processing inequality for fidelity, and is also a direct consequence of the Cauchy--Schwarz inequality.
\end{remark}

The central analysis in our proof will be the following:
\begin{theorem} \label{thm:subtest}
        With $\calT_x$ being a (DT) basis in which $\ket{\tar^{x0}}$ and $\ket{\lab^{x1}}$ are both phase states, 
        the following holds:
     \begin{align}
         \Pr[\textsc{SubTest}_{\calT_x}(\ket{\lab^x} : \ket{\smash{\tar^x}}) \text{ rejects}] 
         &\geq \Delta(\ket{\lab^x}: \ket{\tar^x}), \label{eqn:reject}\\
         \Pr[\textsc{SubTest}_{\calT_x}(\ket{\lab^x} : \ket{\tar^x}) \text{ accepts}]
         &\geq \abs{\braket{\tar^x}{\lab^x}}^2. \label{eqn:accept}
     \end{align}
\end{theorem}
We defer the proof of \Cref{thm:subtest} to \Cref{sec:subtest} and conclude assuming it holds. For now, we use \Cref{thm:subtest} to complete the analysis of our algorithm $\textsc{Certify}$. 

\begin{proposition}
    It holds that
    \[
        \mathop{\E}_{x} \sbra{ \Delta(\ket{\lab^x}: \ket{\tar^x}) } = \frac1n\cdot \left(1- \abs{\braket{\tar}{\lab}}^2 \right)
    \]
    Here, $x$ is sampled by performing steps 1--2 of our algorithm. In other words, $x$ is sampled by first choosing a uniformly random $k \in [n]$ and then measuring the first $k-1$ qubits of $\ket{\lab}$ to obtain $x$. 
\end{proposition}

\begin{proof}
    For $0\leq k \leq n$, define the quantity
    \begin{equation*}
        \Phi^k \coloneqq \mathop{\E}_{x} \left[\abs{\braket{\tar^{x}}{\lab^{x}}}^2\right].
    \end{equation*}
    where $x$ is sampled by measuring the first $k$ qubits of $\ket{\lab}$. Then,
    \begin{align*}
        \mathop{\E}_{x} \sbra{ \Delta(\ket{\lab^x}: \ket{\tar^x}) } &= \mathop{\E}_{k\in [n]} \sbra{ \mathop{\E}_{\substack{x, \\ |x| = k-1}}  \mathop{\E}_{b}\left[\abs{\braket{\tar^{xb}}{\lab^{xb}}}^2\right] -  \mathop{\E}_{\substack{x, \\ |x| = k-1}} \sbra{ \abs{\braket{\tar^x}{\lab^x}}^2. }}\\
        &= \mathop{\E}_{k\in [n]} \sbra{ \mathop{\E}_{\substack{x, \\ |x| = k}}  \left[\abs{\braket{\tar^{xb}}{\lab^{xb}}}^2\right] -  \mathop{\E}_{\substack{x, \\ |x| = k-1}} \sbra{ \abs{\braket{\tar^x}{\lab^x}}^2. }}\\
        &= \mathop{\E}_{k\in [n]} \sbra{ \Phi^k - \Phi^{k-1} } \\
        &= \frac{1}{n}\cdot (\Phi^n-\Phi^0) = \frac1n\cdot \left(1- \abs{\braket{\tar}{\lab}}^2 \right),
    \end{align*}
    where the last line follows by a telescoping sum.
\end{proof}

It follows from \Cref{thm:subtest} that
\begin{align*}
    \Pr[\textsc{Certify} \text{ rejects}] &= \mathop{\E}_{x} \Pr[\textsc{SubTest}_{\calT_x}(\ket{\lab^x} : \ket{\smash{\tar^x}}) \text{ rejects}]  \\
    &\geq  \mathop{\E}_{x}\ [\Delta(\ket{\lab^{x}}: \ket{\tar^{x}})]\\
    &= \frac1n\cdot \left(1- \abs{\braket{\tar}{\lab}}^2 \right)
\end{align*}
as claimed. As for the acceptance probability,
\begin{align*}
    \Pr[\textsc{Certify} \text{ accepts}] &= \mathop{\E}_{x} \Pr[\textsc{SubTest}_{\calT_x}(\ket{\lab^x}: \ket{\smash{\tar^x}}) \text{ accepts}]  \\
    &\geq \mathop{\E}_{x} \sbra{ \abs{\braket{\tar^x}{\lab^x}}^2 }\\
    &\geq \abs{\braket{\tar}{\lab}}^2,
\end{align*}
where the last inequality follows by \Cref{rem:nonneg}. This concludes the proof of \Cref{thm:main}.

\subsection{Analysis of the Subtest (Proof of \Cref{thm:subtest})} \label{sec:subtest}

We will analyze $\textsc{SubTest}_{\calT}(\ket{V}: \ket{U})$ for any two $m$ qubit states $\ket{V}$ and $\ket{U}$ and any $m-1$ qubit orthonormal basis $\calL$ in which $\ket{U^0}$ and $\ket{U^1}$ are both phase states. (Here, $\ket{U^b}$ denotes the reduced state if one were to measure the first qubit of $\ket{U}$ and observe $b$.) We denote the basis vectors of $\calL$ by $\ket{\ell}$, where $\ell \in [2^{m-1}]$.
Write
\begin{equation*}
    \ket{U} = \ket{0}\otimes u^0 + \ket{1}\otimes u^1, \qquad 
    \ket{V} = \ket{0}\otimes v^0 + \ket{1}\otimes v^1.
\end{equation*}
Since both $\ket{U^0}$ and $\ket{U^1}$ are phase states in the basis $\calL$, there is a unitary $D$ that is diagonal in $\calL$ such that $\ket{U^1} = D\ket{U^0}$. Define $\widetilde{v}^1 \coloneqq D^\dagger v^1$. Also, for ease of notation let $a^0=\norm{\smash{u^0}}$ and $a^1 = \norm{\smash{u^1}}$.
\begin{proposition} \label{prop:subtest}
    The following two inequalities hold:
    \begin{align} 
        \Pr[\textsc{SubTest} \text{ accepts}] &= \left\|a^0 v^0 + a^1 \widetilde{v}^1\right\|^2. \label{eqn:acc0} \\
        \Pr[\textsc{SubTest} \text{ rejects}] &= \left\|a^1 v^0 - a^0 \widetilde{v}^1\right\|^2. \label{eqn:start}
    \end{align}
\end{proposition}

\begin{proof}
We will first show \Cref{eqn:acc0}. Write $v^0= \sum_\ell v^0_\ell \ket{\ell}$ and similarly decompose $v^1$ and $\widetilde{v}^1$. Denote the first qubit of $\ket{V}$ conditioned on observing $\ket{\ell}$ by $ \ket{V_\ell}$, and the first qubit of $\ket{U}$ conditioned on observing $\ket{\ell}$ by $ \ket{U_\ell} \coloneqq a^0\ket{0} + a^1\zeta_\ell \ket{1}$, where $\zeta_\ell$ denotes the diagonal entry of $D$ corresponding to the basis vector $\ket{\ell}$. Then, the probability of the subtest  accepting is the expected fidelity of $\ket{U_\ell}$ and $\ket{V_\ell}$ upon measuring $\ket{V}$:
\[
     \mathop{\E}_\ell \abs{ \braket{U_\ell}{V_\ell} }^2 = \sum_\ell \abs{ \bra{U_\ell}(v^0_\ell\ket{0} + v^1_\ell\ket{1}) }^2 = \sum_\ell \abs{ a^0v^0_\ell + a^1 \overline{\zeta_\ell} v^1_\ell\ket{1}) }^2 =  \sum_\ell \abs{ a^0v^0_\ell + a^1 \widetilde{v}^1_\ell }^2,
\]
where the last equality uses that $\widetilde{v}^1_\ell = \overline{\zeta_\ell}v^1_\ell$. Using the Pythagorean theorem now yields the claimed \Cref{eqn:acc0}.

To show \Cref{eqn:start}, it suffices to prove that the claimed rejection probability is 1 minus the above. Indeed, by expanding, we see that
\begin{align*}
    \left\|a^0 v^0 + a^1 \widetilde{v}^1\right\|^2 + \left\|a^1 v^0 - a^0 \widetilde{v}^1\right\|^2 &= \left( (a^0)^2+(a^1)^2 \right) \left( \| v^0 \|^2 + \| \widetilde{v}^1 \|^2 \right) \\
    &\quad~~ + \left( a^0a^1-a^1a^0 \right) \left( v^{0\dag} \widetilde{v}^1 + \widetilde{v}^{1\dag} v^0 \right)\\
    &= 1,
\end{align*}
where the last line holds because both $(a^0)^2+(a^1)^2$ and $\| v^0 \|^2 + \| \widetilde{v}^1 \|^2$ are equal to 1.
\end{proof}

    We now use these expressions to establish \Cref{thm:subtest}. Note that we have $\ket{V^b} = \frac{v^b}{\|v^b\|}$ and $\ket{U^b} = \frac{u^b}{\norm{u^b}}$ for $b = 0, 1$.
    Thus
    \begin{align*}
        \abs{\braket{U}{V}}^2 
        &= \left|a^0\cdot \bra{U^0}v^0 + a^1\cdot \bra{U^1}v^1\right|^2 
        = \left|a^0\cdot \bra{U^0}v^0 + a^1\cdot \bra{U^0}\widetilde{v}^1\right|^2\\
        &= \left|\bra{U^0}\left(a^0 v^0 + a^1 \widetilde{v}^1\right)\right|^2
        \leq \left\|a^0 v^0+a^1 \widetilde{v}^1\right\|^2, \label{eqn:expre}
    \end{align*}
    where we used $\bra{U^1}v^1 = \bra{U^0}D^\dagger D \widetilde{v}^1 = \bra{U^0}\widetilde{v}^1$.
    In combination with \Cref{eqn:acc0}, this shows the subtest acceptance probability is at least the fidelity $\abs{\braket{U}{V}}^2$, confirming \Cref{eqn:accept}. 
    
    To show \Cref{eqn:reject} it will be helpful to re-express $\Delta(\ket{V} : \ket{U})$:
    \begin{proposition}
        It holds that
        \[
            \Delta(\ket{V} : \ket{U}) = \left|a^1 \bra{U^0}v^0 - a^0 \bra{U^1}v^1\right|^2
        \]
    \end{proposition}
    \begin{proof}
        This follows by expanding:
        \begin{align*}
            \Delta(\ket{V} : \ket{U}) &= \|v^0\|^2 \cdot \left|\braket{U^0}{V^0}\right|^2 + \|v_1\|^2 \cdot \left|\braket{U^1}{V^1}\right|^2  
            - \left|\braket{U}{V}\right|^2 \\
            &= \left|\bra{U^0}v^0\right|^2 + \left|\bra{U^1}v^1\right|^2 - \left|a^0 \bra{U^0}v^0 + a^1 \bra{U^1}v^1\right|^2 \\
            &= \left( 1- (a^0)^2 \right) \left|\bra{U^0}v^0\right|^2 + \left( 1- (a^1)^2 \right) \left|\bra{U^1}v^1\right|^2 - 2a^0a^1\Re \left( v^0\ket{U^0}\bra{U^1}v^1 \right) \\
            &= (a^1)^2\left|\bra{U^0}v^0\right|^2 + (a^0)^2 \left|\bra{U^1}v^1\right|^2 - 2a^0a^1\Re \left( v^0\ket{U^0}\bra{U^1}v^1 \right) \\
            &=\left|a^1 \bra{U^0}v^0 - a^0 \bra{U^1}v^1\right|^2 \qedhere
        \end{align*}
    \end{proof}
    Then, using that $\bra{U^1}v^1 = \bra{U^0}\widetilde{v}^1$ gives
    \[
        \Delta(\ket{V} : \ket{U}) = \left|a^1 \bra{U^0}v^0 - a^0 \bra{U^1}v^1\right|^2 = \left|\bra{U^0}\left(a^1 v^0-a^0 \widetilde{v}^1\right)\right|^2 \leq \left\|a^1 v^0 - a^0 \widetilde{v}^1\right\|^2.
    \]
    In combination with \Cref{eqn:start}, this shows the subtest rejection probability is at least $\Delta(\ket{V} : \ket{U})$, confirming \Cref{eqn:reject}.

\subsection{Access Model and Computational Efficiency}\label{sec:access model}
Thus far we have not constrained how the target state $\ket{\tar}$ is represented classically and granted ourselves unlimited classical computation to compute any properties of it. Hence, while our procedure is quantum-efficient, it is not yet classically efficient.

Because an explicit classical description of $\ket{\tar}$ can be exponentially large, we should not allow ourselves to work with a full description of $\ket{\tar}$. Instead, we work in an \emph{oracle-access model} mirroring Huang et al.~\cite{huang2024certifying}. In that work, the oracle provides the value of $\braket{x}{\tar}$ for any queried computational‑basis string $x\in\{0,1\}^{n}$.  We adopt a natural extension: for any sequence of single-qubit measurements and outcomes, the oracle tells us the probability of obtaining that specific outcome sequence:

\begin{definition}\label{def:access model}
We consider oracle access to a state
$\ket{\psi}$  supporting the following type of query: For any operator $\Pi=\Pi_1\otimes\dots\otimes \Pi_n$ that is a tensor product of single-qubit projectors on $\mathbb{C}^{2}$, the querying with $\Pi$ returns the value $\bra{\psi} \Pi \ket{\psi}$. Note that it is possible for individual tensor factors to be $\Pi_k=\Id$.
\end{definition}

It is useful to observe, for example, that if we had a description of $\ket{\tar}$ as a matrix product state of small bond dimension, we could efficiently emulate the oracle.

Now, we will show that this access model allows us to adaptively compute a DT basis as in \Cref{cor:perpy} and implement the corresponding measurements efficiently.

\begin{lemma}\label{lem:computational efficiency}
    Given access to $\ket{\smash{\psi^{ 0}}}$ and $\ket{\smash{\psi^{ 1}}}$ via this model, we can implement measuring in the DT basis defined in \Cref{cor:perpy} (making $\ket{\smash{\psi^0}}$ and $\ket{\smash{\psi^1}}$ phase states), using $O(n)$ time and single-qubit measurements.
\end{lemma}
\begin{proof}
    The algorithm will initialize $t=0$, $\ket{\smash{\psi^0_t}}=\ket{\smash{\psi^0}}$, $\ket{\smash{\psi^1_t}}=\ket{\smash{\psi^1}}$, and an initially empty tensor product of projections $\Pi_t=1$ that corresponds to the measurement outcomes observed up to the $t$th single-qubit measurement. While $t<n$, we compute the reduced states $\rho^0_t = \Tr_{[n]\setminus \{t+1\}}(\ketbra{\smash{\psi^0_t}}{\smash{\psi^0_t}})$ and $\rho^0_t = \Tr_{[n]\setminus \{t+1\}}(\ketbra{\smash{\psi^1_t}}{\smash{\psi^1_t}})$. This can be done by querying the oracle for the quantities
    \begin{align*}
        \bra{\smash{\psi^0}}\Pi\ket{\smash{\psi^0}}\text{ and }\bra{\smash{\psi^1}}\Pi\ket{\smash{\psi^1}},\text{ where }\Pi = \Pi_t \otimes \ketbra{b}{b} \otimes \Id^{\otimes (n - t - 1)}.
    \end{align*}
    for $b\in\{0,+,\mathfrak{i}\}$ and performing a single-qubit quantum state reconstruction algorithm. 
    
    As in \Cref{cor:perpy}, let $\ket{e_t}$ be perpendicular to both $\rho^0$ and $\rho^1$ on the Bloch sphere and measure the the $t$th qubit in the basis $\ket{e_t},\ket{\smash{e_t^\perp}}$. Then set $\Pi_{t+1}=\Pi_t \otimes\ketbra{b}{b}\otimes\Id$, where $\ket{b}\in \{\ket{e_t},\ket{\smash{e^\perp_t}}\}$ was the outcome of this measurement. Then $\ket{\smash{\psi^0_{t+1}}}\propto\Pi_{t+1}\ket{\smash{\psi^0}}$ and $\ket{\smash{\psi^1_{t+1}}}\propto\Pi_{t+1}\ket{\smash{\psi^1}}$ and we increase $t$ by $1$. Repeating for all $t$ gives a measurement outcome in a DT basis satisfying the conclusion of \Cref{cor:perpy}. Computing the $t$th measurement and updating the query take $O(1)$ time, so the overall runtime is $O(n)$.
\end{proof}
Given \Cref{lem:computational efficiency}, we see that \Cref{alg:main} can be implemented in time $O(n)$ under the access model of \Cref{def:access model}. For measurement in the decision tree basis $\calT_x$ (where $x$ is the outcome of measuring the first $k-1$ qubits), we use \Cref{lem:computational efficiency} with $\ket{\psi^0}=\ket{\tar^{x0}}$ and $\ket{\psi^1}=\ket{\tar^{x1}}$.

\section{Lower Bound}\label{sec:lower bound}
In this section, we prove \Cref{prop:lower bound}.
A simple geometric fact is at heart of our lower bound proof:
\begin{claim}\label{claim:uncertainty one particle}
    Set $\ket{\chi_0}, \ket{\chi_1}, \ket{\chi_2}, \ket{\chi_3}$ to be as in the SIC-POVM. Then for all $1$-qubit bases $\{\ket{\phi}, \ket{\phi^\bot}\}$, there exists $b \in \{1,2,3\}$ such that
    \begin{equation*}
        \abs{\braket{\phi}{\chi_0}\braket{\chi_b}{\phi}}+ \abs{\braket{\smash{\phi^\perp}}{\chi_0}\braket{\chi_b}{\smash{\phi^\perp}}} \leq 0.99
    \end{equation*}
    for at least one of $b\in\{1,2,3\}$.
\end{claim}
\begin{proof}
    The SIC-POVM vectors form a tetrahedron on the Bloch sphere. These vectors are far from coplanar, so the result follows.
\end{proof}

\subsection{The Hard-to-Certify State}
\begin{definition}
    For $C \coloneqq (c^1\ldots c^N) \in \{0,1,2,3\}^n$, define the vector 
    \begin{align*}
        v_C\coloneq& \frac1{\sqrt{N}} \sum_{t\in[N]} \ket{\smash{\chi_{c^t_1}}}\otimes\dots\otimes \ket{\smash{\chi_{c^t_n}}}
    \end{align*}
    Define $\ket{\psi_C}$ to be the quantum state in the direction of $v_C$, and let $\ket{\psi_{c^t}}\coloneq \ket{\smash{\chi_{c^t_1}}}\otimes\dots\otimes \ket{\smash{\chi_{c^t_n}}}$.
\end{definition}

\begin{lemma}\label{lem:almost state}
    Let $N=\floor{\smash{2^{10^{-10}n}}}$ and choose $C \coloneqq (c^1\ldots c^N )\in \{0,1,2,3\}^n$ i.i.d. With probability at least 0.9,
    \[
        1 - 2^{-0.1n} \leq \norm{v_C}^2 \leq 1 + 2^{-0.1n}.
    \]
\end{lemma}
\begin{proof}
    By a standard coding theory argument, we have that with probability at least $0.9$, for all $s\neq t$, it holds that $c^s$ and $c^t$ have distance at least $0.1n$. For any $i$ such that $c^s_i \neq c^t_i$, it holds that $\big|\braket{\smash{\chi_{c^s_i}}}{\smash{\chi_{c^t_i}}} \big| \leq \frac13$. Then, when we evaluate $\| v_C \|^2$, each of the $\binom{N}{2}$ cross terms is at most $\frac{1}{N}\cdot \big( \frac{1}{3} \big)^{0.1n}$. As such, the total contribution of the cross terms is at most $\binom{N}{2} \cdot \frac{1}{N}\cdot \big( \frac{1}{3} \big)^{0.1n} \leq 2^{-0.1n}$, and moreover, the diagonal terms sum to $1$.
\end{proof}

Therefore, $v_C$ is extremely close to being a quantum state $\ket{\psi_C}$. In the next section, we will show that it is also likely that $\ketbra{\psi_C}{\psi_C}$ is indistinguishable from the mixed state
\begin{align*}
    \rho_C&\coloneq \frac1{N}\sum_{s}\ketbra{\psi_{c^s}}{\psi_{c^s}}.
\end{align*}

\subsection{The Cross Term Operators}

We will show that the advantage of distinguishing between $\ketbra{\psi_C}{\psi_C}$ and the mixed state $\rho_C$ is related to sum of the absolute values on the diagonal of the operator $\frac1{N}\sum_{s\neq t}\ketbra{\psi_{c^s}}{\psi_{c^t}}$, maximized over the set of product bases for $(\C^2)^{\otimes n}$. More formally we have:

\begin{lemma}\label{lem:distinguishing probability operator diff meghal}
    Let $\ket{\phi_x}$ with $x\in\{0,1\}^n$ form an orthonormal basis for $(\C^2)^{\otimes n}$. Suppose that $C$ of size $N=\floor{\smash{2^{10^{-10}n}}}$ satisfies the normalization condition of \Cref{lem:almost state}. Then for large enough $n$,
    \begin{align*}
        d_{\mathrm{TV}}(\calD_1,\calD_2) &\leq \frac1{N}\sum_{x}\abs{\sum_{s\neq t}\braket{\phi_x}{\smash{\psi_{c^s}}}\braket{\smash{\psi_{c^t}}}{\phi_x}} + 2^{-0.00 01n}.
    \end{align*}
\end{lemma}
\begin{proof}
    We directly compute
    \begin{align*}
        2d_{\mathrm{TV}}(\calD_1,\calD_2) &= \sum_x\abs{\braket{\phi_x}{\psi_C}\braket{\psi_C}{\phi_x}-\bra{\phi_x}\rho_C\ket{\phi_x}}=\sum_x\abs{{\bra{\phi_x}(\ketbra{\psi_C}{\psi_C}-\rho_C)\ket{\phi_x}}}.
    \end{align*}
    By \Cref{lem:almost state}, we can approximate the matrix in the middle up to spectral norm:
    \begin{align*}
        \ketbra{\psi_C}{\psi_C}-\rho_C \underset{N/2^{0.001n}}{\approx}v_Cv_C^\dag -\rho_C=\frac1{N}\sum_{s,t}\ketbra{\smash{\psi_{c^s}}}{\smash{\psi_{c^t}}}-\frac1{N}\sum_{s}\ketbra{\smash{\psi_{c^s}}}{\smash{\psi_{c^s}}}=\frac1{N}\sum_{s\neq t}\ketbra{\smash{\psi_{c^s}}}{\smash{\psi_{c^t}}}.
    \end{align*}
    Plugging in the value of $N$, we conclude the result when $n$ is large enough.
\end{proof}

For $S\subseteq[n]$, we say that $\{\ket{\phi_x}\}$ form an $\ell$-adaptive orthonormal basis for $(\C^2)^{\otimes n}$ if there exists $S\subseteq[n]$ of size $\leq\ell$ such that each basis element $\ket{\phi_x}$ takes the form 
\begin{align*}
    \ket{\phi_x} &= \bigotimes_{i\not\in S}\ket{\smash{\phi^{i}_{x_i}}} \otimes \ket{\phi_{x_S}}. 
\end{align*}
Measurement in such a basis can be done by measuring all of the qubits not in $S$ in a non-adaptive manner, and then measuring the qubits in $S$ in some basis depending on the outcome of the first $n-\ell$ measurements.

To show that indeed a random $C$ will give a hard to distinguish state, we show the following:
\begin{lemma}\label{lem:lower bound from concentration}
    Let $C\coloneq (c^1\ldots c^N)$, where each $c^s$ is chosen independently  and uniformly at random from  $\{0,1,2,3\}^n$, and where $N\coloneq \floor{\smash{2^{10^{-10}n}}}$. Then if $n$ is large enough, with probability at least 0.9 we have for all $10^{-8}n$-adaptive bases $\{\ket{\phi_x}\}$ that
    \begin{align*}
        \sum_{x\in\{0,1\}^n}\abs{\sum_{s\neq t\in[N]}\braket{\phi_x}{\smash{\psi_{c^s}}}\braket{\smash{\psi_{c^t}}}{\phi_x}}\leq 5.
    \end{align*}
\end{lemma}
\begin{proof}
    Since the $\ket{\phi_x}$ form an $10^{-8}$-adaptive basis, there exists $S\subseteq[n]$ of size at most $10^{-8}n$ such that we can use the following expansion:
    \begin{align*}
        &\sum_{x\in\{0,1\}^n}\abs{\sum_{s\neq t}\braket{\phi_x}{\smash{\psi_{c^s}}}\braket{\smash{\psi_{c^t}}}{\phi_x}}\leq \sum_{x\in\{0,1\}^n}\sum_{s\neq t}\abs{\braket{\phi_x}{\smash{\psi_{c^s}}}\braket{\smash{\psi_{c^t}}}{\phi_x}}\\
        &\leq\sum_{s\neq t}\sum_{x\in\{0,1\}^n}\prod_{i\not\in S}\abs{\braket{\smash{\phi^i_{x_i}}}{\chi_{c^s_i}}\braket{\chi_{c^t_i}}{\smash{\phi^i_{x_i}}}}= \sum_{s\neq t}\prod_{i\not\in S}\pbra{\abs{\braket{\smash{\phi^i_{0}}}{\chi_{c^s_i}}\braket{\chi_{c^t_i}}{\smash{\phi^i_{0}}}}+\abs{\braket{\smash{\phi^i_{1}}}{\chi_{c^s_i}}\braket{\chi_{c^t_i}}{\smash{\phi^i_{1}}}}}.
    \end{align*}
    Notice that the term inside the right-hand side expression is always at most $1$ and sometimes less than $1$. Our goal will be to show that or all product bases $\phi$, for most pairs $(s,t)$, the product will be close to 0. 
    
    For any state one qubit basis $\phi'$, denote by $b(\phi')$ the element in $\{1,2,3\}$ such that the inequality in \Cref{claim:uncertainty one particle} holds (any such $b$ if there are multiple). Say a pair $(s,t)$, is \emph{bad} for an $\ell$-adaptive basis $\{\ket{\phi_x}\}$ (with corresponding set $S$) if for at least $10^{-6}n$ indices $i\in [n]\setminus S$, it holds that $\{c^s_i,c^t_i\} = \{0, b(\phi^i)\}$. Note that if the pair $(s,t)$ is bad for such a basis $\{\ket{\phi_x}\}$, then
    \[
        \prod_{i\not\in S}\pbra{\abs{\braket{\smash{\phi^i_{0}}}{\chi_{c^s_i}}\braket{\chi_{c^t_i}}{\smash{\phi^i_{0}}}}+\abs{\braket{\smash{\phi^i_{1}}}{\chi_{c^s_i}}\braket{\chi_{c^t_i}}{\smash{\phi^i_{1}}}}} \leq 0.99^{10^{-6}n}.
    \]
    Consider any quintuple of distinct pairs whose elements are in $[N]$ denoted by $\{s_1,t_1\}\ldots \{s_5,t_5\}$. We will show that with probability at least $1-2^{-10^{-5}n}$ at least one pair in the quintuple is bad for any $10^{-8}n$-adaptive basis.
    
    Once we have shown this, taking a union bound over all of the at most $N^{10}$ quintuples of pairs, the probability that $C$ is such that for every $10^{-8}n$-adaptive basis, at least one pair in every quintuple is bad for that basis is at least
    \[
        1-2^{-10^{-5}n}\cdot N^{10} \geq 1-2^{-10^{-6}n}.
    \]
    When this event holds, $C$ will be such that for any $10^{-8}n$-adaptive basis $\{\ket{\phi_x}\}$, at most four pairs $\{s_k,t_k\}$ are not bad for that basis, so
    \[
        \sum_{s\neq t}\prod_{i}\pbra{\abs{\braket{\smash{\phi^i_{0}}}{\chi_{c^s_i}}\braket{\chi_{c^t_i}}{\smash{\phi^i_{0}}}}+\abs{\braket{\smash{\phi^i_{1}}}{\chi_{c^s_i}}\braket{\chi_{c^t_i}}{\smash{\phi^i_{1}}}}} \leq 4+N^2\cdot 0.99^{-10^{-6}n} \leq 5
    \]
    when $n$ is large enough, as desired.
    
    It suffices to show that with probability at least $2^{-10^{-5}n}$, the $c^{s_1},c^{t_1},\dots,c^{t_5}$ are such that for any $10^{-8}n$-adaptive basis, at least one of the pairs $\{s_k,t_k\}$ is bad for that basis. The graph on $[N]$ that the pairs form either has a vertex of degree 3 or consists of disjoint cycles and edges. This means that for a fixed index $i\in[n]$, there is a labeling $(x_{s_1},x_{t_1},\dots,x_{s_5},x_{t_5})\in \{0,1,2,3\}^{10}$ such that if $c^{s_1}_i=x_{s_1},\dots,c^{t_5}_i=x_{t_5}$ then
    \begin{align*}
        \{\{c^{s_1}_i,c^{t_1}_i\},\dots,\{c^{s_5}_i,c^{s_5}_i\}\}=\{\{0,1\},\{0,2\},\{0,3\}\}
    \end{align*}
    If this occurs, then for any choice of $\phi^i$, at least one pair $\{s_k,t_k\}$ will satisfy $\{c^{s_k}_i,c^{t_k}_i\} = \{0, b(\phi^i)\}$. Moreover, the probability of this labeling occurring among the three chosen pairs is at least $4^{-6}$.

    Now let $w(\{s_1,t_1\},\dots \{s_5,t_5\})$ be the number of indices $i$ on which this happened. By a Chernoff bound, we have
    \begin{align*}
        \Pr\sbra{w(\{s_1,t_1\},\dots \{s_5,t_5\}) \leq \frac{4^{-6}n}{4}} \leq 2^{-10^{-5}n}.
    \end{align*}
    When this happens, there must be some pair $\{s_k,t_k\}$ where the condition holds for at least $4^{-6}n/20>10^{-5}n$ indices. Such a pair is bad for any $10^{-8}n$-adaptive basis, since $10^{-5}n-10^{-8}n\geq 10^{-6}n$. Thus, at least one of the five pairs is bad.
\end{proof}

\subsection{The Lower Bound}
We complete this section by proving our lower bound:
\begin{proof}[Proof of \Cref{prop:lower bound}]
    Let $C$ be the collection of size $N=2^{10^{-10}n}$ which exists by \Cref{lem:lower bound from concentration}. Any $10^{-8}n$-adaptive algorithm is one that measures each copy of a state in an $10^{-8}$-adaptive basis and makes decisions based on these measurement outcomes. By \Cref{lem:distinguishing probability operator diff meghal,lem:lower bound from concentration}, the states $\ketbra{\psi_C}{\psi_C}$ and $\rho_C$ yield outcome distributions that have TV distance $\leq 2^{-\Omega(n)}$ in any $10^{-8}$-adaptive basis. By a standard coupling argument, this shows that if $\mathfrak{D}_1$ and $\mathfrak{D}_2$ are the distributions corresponding to measuring $2^{o(n)}$ copies of $\ketbra{\psi_C}{\psi_C}$ and $\rho_C$, respectively, in a sequence of potentially adaptive product bases, then $d_{\mathrm{TV}}(\mathfrak{D}_1,\mathfrak{D}_2)\leq 2^{-\omega(n)}$.
    
    On the other hand, a similar calculation as that in the proof of \Cref{lem:almost state} we have that $\bra{\psi_C}\rho_{C}\ket{\psi_C}\leq 2^{-\Omega(n)}$, so this yields a lower bound for certifying $\ket{\psi_C}$.
\end{proof}

\section*{Acknowledgments}
We thank Mihir Singhal for helpful discussions, Omar Alrabiah for checking calculations, Mingyu Sun for pointing out errors in an earlier version of the paper, and Sisi Zhou for the reference~\cite{zhou2020saturating}. We also thank ChatGPT and Gemini for assistance. MG and WH are grateful to Angelos Pelecanos for his daily presence in Soda Hall.

\bibliography{references}
\bibliographystyle{alpha}

\end{document}